\documentclass[runningheads,oribibl,a4paper]{llncs}

\usepackage{amsmath}
\usepackage{amssymb}
\usepackage{graphicx}
\usepackage{subfigure} 

\usepackage{url}
\usepackage{enumerate}

\usepackage{multirow}

\renewcommand{\bbbn}{\mathbb{N}}

\newcommand{\setcompr}[2]{\ensuremath{\{#1\ |\ #2\}}}
\newcommand{\lclause}[2]{\ensuremath{#1\,||\,#2}}

\newcommand{\emptycl}{\ensuremath{\bot}}

\newcommand{\buchi}{B\"uchi}

\newcommand{\ltlnext}{\ensuremath{\bigcirc}}    
\newcommand{\always}{\ensuremath{\Box}}         
\newcommand{\eventually}{\ensuremath{\Diamond}} 
\newcommand{\until}{\ensuremath{\mathsf{U}}}
\newcommand{\release}{\ensuremath{\mathsf{R}}}

\newcommand{\hidden}[1]{}

\begin{document}

\mainmatter

\title{Variable and clause elimination for \\
LTL satisfiability checking\thanks{Partly supported by Microsoft Research through its PhD Scholarship Programme.}}

\author{Martin Suda
}

\institute{
Max-Planck-Institut f\"ur Informatik, Saarbr\"ucken, Germany
\\ 
Saarland University, Saarbr\"ucken, Germany
\\ 
Charles University, Prague, Czech Republic
}

\maketitle


\begin{abstract}
We study preprocessing techniques for clause normal forms of LTL formulas.
Applying the mechanism of labelled clauses
enables us to reinterpret LTL satisfiability as a set of purely propositional problems
and thus to transfer simplification ideas from SAT to LTL.
We demonstrate this by adapting variable and clause elimination,
a very effective preprocessing technique used by modern SAT solvers.
Our experiments confirm that even in the temporal setting substantial
reductions in formula size and subsequent decrease of solver runtime can be achieved.
\end{abstract}

\section{Introduction}

Linear temporal logic (LTL) is a modal logic with modalities referring to time \cite{firstP77}.
Traditionally, it finds its use in formal verification of reactive systems
where it serves as a specification language for expressing the system's desired behavior.
The specifications are subsequently checked against a model of the system
during the process of \emph{model checking} \cite{mcCGPbook01}.
More recently, the importance of LTL \emph{satisfiability checking} is becoming recognized 
\cite{automRozVar07,Schuppan11}, where the task is to decide whether a given LTL formula has a model at all.
This is, for instance, essential for assuring quality of formal specifications \cite{fahrPill06}.
Satisfiability checking of LTL is a computationally difficult task, in fact a PSPACE-complete one \cite{cplxSC85}, 
and thus techniques for improving solving methods are of practical importance. 


One possibility for speeding up the checking lies in simplifying the input formula
before the actual decision method is started. 
In the context of resolution-based methods for LTL satisfiability 
\cite{ctrFDP01,LS4PLTLpaper12}, on which we focus here,
formulas are first translated into a clause normal form.
Simplification then means reducing the number of clauses and variables
 while preserving satisfiability of the formula.
Such a preprocessing step may have a significant positive impact on the subsequent running time.

In this paper we take inspiration from the SAT community where a technique 
called variable and clause elimination \cite{EB05} has been shown to be particularly effective.
It combines exhaustive application of the resolution rule over selected variables
with subsumption and other reductions. Our main contribution lies in showing
that variable and clause elimination can be adapted from SAT to the setting of LTL.
This is quite non-trivial, because LTL normal forms consist of \emph{temporal clauses},
which are bound to specific temporal contexts and so their interactions in inferences and reductions
need to be carefully controlled.

A general method for reducing LTL satisfiability to the purely propositional setting has been introduced in \cite{LS4PLTLpaper12}.
There, the existence of a model of an LTL formula is shown to be equivalent to satisfiability 
of one of infinitely many potentially infinite standard clause sets.
These are, however, finitely represented with the help of \emph{labels}, which allows for
an effective transfer of resolution-based reasoning techniques from propositional logic to LTL.
In this paper, we extend the ideas of \cite{LS4PLTLpaper12} to adapt variable and clause elimination.
An additional label component is needed to justify elimination in its general form,
but we prove it can be dispensed with after the elimination process.



Our exposition starts in Sect.~\ref{sec_normal},
where we describe our version of clause
normal form of LTL formulas, which we call LTL-specification.
Specifications are a particular refinement of the Separated Normal Form \cite{fisher91},
which can be seen as concise descriptions of B\"uchi automata.
This observation, which is of independent interest, represents another contribution of this paper.
The mechanism of labelled clauses itself is introduced in Sect.~\ref{sec_label}
and utilized for variable and clause elimination in Sect.~\ref{sec_elim}.
Practical potential of our method is demonstrated in Sect.~\ref{sec_experiment},
where we describe the effect of the simplification on runtimes of two resolution-based LTL provers
over an extensive set of benchmark problems. In Sect.~\ref{sec_discuss}
we follow the connection to B\"uchi automata to discuss related work,
and we conclude in Sect.~\ref{sec_concl} by mentioning possibilities for future work.

\section{Preliminaries} \label{sec_normal}

We assume the reader is familiar with propositional logic and the syntax and semantics of LTL.\footnote{See Appendix \ref{sec_prelim} for a short overview.}
LTL formulas are built over a given \emph{signature} $\Sigma = \{p,q,r,\ldots\}$ of propositional variables
using propositional connectives $\neg, \land, \lor ,\ldots$, and temporal operators $\ltlnext, \always, \eventually, \until, \ldots$ 
Propositional clauses, denoted $C,D$, possibly with subscripts, are sets of literals understood as disjunctions.
A propositional \emph{valuation} is a mapping $W : \Sigma \rightarrow \{0,1\}$.
We write $W \models C$ if a valuation $W$ propositionally satisfies a clause $C$.
An \emph{interpretation} of an LTL formula is an infinite sequence of valuations $(W_i)_{i \in \bbbn}$,
in this context also referred to as \emph{states}.




In order to talk about two neighboring states at once
we introduce a disjoint copy of the basic signature $\Sigma' = \{p',q',r',\ldots\}$.
Given a clause $C$ over $\Sigma$, we write $C'$ to denote its obvious counterpart over $\Sigma'$.
For a valuation $W$ over $\Sigma$ let $W'$ denote the valuation over $\Sigma'$ that behaves on primed symbols
in the same way as $W$ does on unprimed ones.
We therefore have $W\models C$ if and only if $W' \models C'$ for any such $W$ and $C$.
If $W_1$ and $W_2$ are two valuations over $\Sigma$,
we let $[W_1,W_2]$ denote the joined valuation $W_1 \cup (W_2)': \Sigma \cup \Sigma' \rightarrow \{0,1\}$. 
Such a valuation is needed to evaluate clauses over the joined signature $\Sigma \cup \Sigma'$.


           

Most resolution-based approaches to satisfiability checking first translate
the input formula into a certain normal form. In the context of LTL,
the Separated Normal Form (SNF) developed by Fisher \cite{fisher91}
has proven to be very useful. It is obtained from an LTL formula 
by applying transformations that 
1) introduce new variables as names for complex subformulas, 
2) remove temporal operators by expanding their fixpoint definitions,
3) apply classical style rewrite operations to obtain a result which is clausal, 
i.e. represented by a top-level conjunction of temporal clauses, which are disjunctive in nature. 
The whole transformation preserves satisfiability 
of the input formula and it is ensured that the result does not grow in size by more than a linear factor \cite{ctrFDP01}.\footnote{
A streamlined version of the transformation can be found in Appendix \ref{sec_ltl2snf}.}

In this paper we use a particular refinement of SNF which we call LTL-specification \cite{LS4PLTLpaper12}.
To obtain a specification, a general SNF is first normalized further by using the ideas of \cite{simplifiedDFK02}.
In particular, we transform the so called conditional eventuality clauses to unconditional ones 
and then reduce the potentially multiple (unconditional) eventuality clauses to just one eventuality clause.\footnote{
A recapitulation of these refinements has been moved to Appendix \ref{sec_snf2spec}.}
Finally, to obtain a compact representation, we explicitly sort the clauses into three categories,
strip them off the temporal operators and write them down using standard propositional clauses instead.
The semantics is preserved as it now follows from the context.
Even after these refinements the result is linearly bounded in size and equisatisfiable with respect to the original formula.



\begin{definition}
An \emph{LTL-specification} is a quadruple $\mathcal{S} = (\Sigma,I,T,G)$ such that 
	\begin{itemize}
	\item
		$\Sigma$ is a finite propositional signature,
	\item
		$I$ is a set of \emph{initial} clauses $C_i$ over the signature $\Sigma$,
	\item
		$T$ is a set of \emph{step} clauses $C_t \lor (D_t)'$ over the joined signature $\Sigma \cup \Sigma'$,		
	\item
		$G$ is a set of \emph{goal} clauses $C_g$ over the signature $\Sigma$.
	\end{itemize}
\end{definition}
The initial and step clauses are directly translated from SNF. 
The goal clauses \emph{all together} express the single eventuality obtained in the previous step.
This generalization (from a single goal clause) is for free and appears to make the definition conceptually cleaner.
Intuitively, specification stands for the LTL formula
\[\left( \bigwedge C_i \right) \land \Box \left( \bigwedge (C_t \lor \bigcirc D_t) \right) \land \Box \Diamond \left( \bigwedge C_g \right)\enspace ,\]
which directly translates to the following formal definition.
\begin{definition}
An interpretation $(W_i)_{i\in\mathbb{N}}$ is a \emph{model} of $\mathcal{S} = (\Sigma,I,T,G)$ if 
\begin{enumerate}
	\item
		for every $C_i \in I$, $W_0 \models C_i$,

	\item
		for every $i \in\mathbb{N}$ and every $C_t \lor (D_t)' \in T$, $[W_i,W_{i+1}] \models C_t \lor (D_t)'$, and

	\item
		there are infinitely many indices $j$ such that for every $C_g \in G$, $W_j \models C_g$.			
\end{enumerate}
An LTL-specification $\mathcal{S}$ is \emph{satisfiable} if it has a model.
\end{definition}

\begin{remark} \label{rem_buchi}
We close this section with an interesting observation relating our approach to LTL satisfiability 
to explicit methods based on automata. It is well known (see e.g. \cite{GPVW95}) that for any LTL formula
$\varphi$ there is a \buchi\ automaton $\mathcal{A}_\varphi$ recognizing models of $\varphi$,
i.e. an automaton that accepts exactly those valuations $(W_i)_{i\in\mathbb{N}}$ that are models of $\varphi$.
The size of such an automaton, i.e. the number of its states, is bounded by $2^{|\varphi|}$, 
where $|\varphi|$ denotes the size of the formula.


Now we can easily interpret an LTL-specification $\mathcal{S}$ as a \emph{symbolic description} of such an automaton.
The states of the automaton are formed by the set $Q = 2^\Sigma$, i.e. the set of all valuations over $\Sigma$, 
its transition function $\delta = \setcompr{(W_1,W_2)}{[W_1,W_2] \models \bigwedge (C_t \lor (D_t)')}$
contains those pairs of valuations that satisfy the step clauses, and its initial and accepting sets 
are defined as $Q_I = \setcompr{W}{W \models \bigwedge C_i}$ and $Q_F = \setcompr{W}{W \models \bigwedge C_g}$, respectively.
It is easy to check that the models of $\mathcal{S}$ are exactly the accepting runs of this automaton.


This way one can view the transformations from an LTL fomula to SNF and further to LTL-specification
as an alternative way of obtaining a \buchi\ automaton for the formula.
Interestingly, it is only the last step, when the automaton is made explicit, that 
incurs the inherent exponential blowup.
\end{remark}

\section{Mechanism of labelled clauses} \label{sec_label}

The purpose of this section is to show that the task of LTL satisfiability
can be reduced to a set of purely propositional SAT problems.
This provides a means for transferring 
the well-known resolution-based reasoning techniques
from the propositional level to that of LTL.
In particular, it will in Sect.~\ref{sec_elim} allow us to transfer
variable and clause elimination.
The reduction from LTL that we present
leaves us with infinitely many propositional problems 
over an infinite signature. Labels are then used to \emph{finitely} represent 
and control clauses within these problems, abbreviating entire clause sets.

Assume we have an LTL-specification $\mathcal{S} = (\Sigma,I,T,G)$ 
and want to decide satisfiability of the formula it represents.
It is a known fact that when considering satisfiability of LTL formulas
attention can be restricted to \emph{ultimately periodic}~\cite{cplxSC85} 
interpretations. These start with a finite sequence of states and then
repeat another finite sequence of states forever. This observation,
which is one of the key ingredients of our approach, motivates the following definition.

\begin{definition}
Let $K \in \bbbn$, and $L \in \bbbn^+ = \bbbn \setminus \{0\}$ be given.
An interpretation $(W_i)_{i\in\bbbn}$ is a \emph{$(K,L)$-model} of $\mathcal{S}=(\Sigma,I,T,G)$ if
\begin{enumerate}
	\item
		for every $C \in I$, $W_0 \models C$,
	\item
		for every $i \in\bbbn$ and every $C \in T$, $[W_i,W_{i+1}] \models C$,
	\item
		for every $i \in\bbbn$ and every $C \in G$, $W_{(K+i \cdot L)} \models C$.
\end{enumerate}
\end{definition}

Satisfiability within a $(K,L)$-model for \emph{some} values of $K$ and $L$ 
corresponds to the original semantics except 
that the condition on the goal clauses to be satisfied in infinitely many states is 
now controlled and we require that these states form 
an arithmetic progression with $K$ as the initial term and $L$ the common difference.
Please consult \cite{SudaWeidenbachReport2012} for a detailed proof of why 
focusing only on $(K,L)$-models
does not change the notion of satisfiability.

For a particular choice of $K$ and $L$, 
the existence of a $(K,L)$-model can be stated as an infinite but purely propositional problem 
over the infinite signature $\Sigma^* = \bigcup_{i\in \bbbn} \Sigma^{(i)}$. 
Here we extend the convention about priming and allow it to be applied more than once.
Thus along with signatures $\Sigma$ and $\Sigma'$ we also have $\Sigma'', \Sigma''', \ldots$ (also written $\Sigma^{(2)}, \Sigma^{(3)}, \ldots$),
as other disjoint copies of the basic signature implicitly meant to represent states further in the future.
Now the purely propositional problem simply restates the definition of a $(K,L)$-model
in the form of clauses over $\Sigma^*$,
making use of the natural bijection between propositional valuations over $\Sigma^*$ and interpretations.\footnote{
Given $W^*: \Sigma^* \rightarrow \{0,1\}$, the corresponding interpretation $(W_i)_{i\in\mathbb{N}} : \bbbn \times \Sigma \rightarrow \{0,1\}$
is defined by the equation $W_i(p) = W^*(p^{(i)})$ for every $i \in \bbbn$ and every $p \in \Sigma$.}
It consists of:
\begin{itemize}
\item
	the set of initial clauses $I = \setcompr{C^{(0)}}{C \in I}$,
\item
	together with $\setcompr{C^{(i)}}{C \in T, i \in \bbbn}$,	
\item
	and with $\setcompr{C^{(K+i \cdot L)}}{C \in G, i \in \bbbn}$,
\end{itemize}
where the symbol $C^{(i)}$ means that each literal in $C$ is being ``moved $i$ signatures forward''.
Thus, e.g., for a clause $C = p \lor q'$ over $\Sigma \cup \Sigma'$ we denote by $C^{(2)}$ 
the clause $p^{(2)} \lor q^{(3)}$ over $\Sigma^{(2)} \cup \Sigma^{(3)}$. See Figure~\ref{fig_klmodel}
for an illustration of the situation.



\begin{figure}[tb]
\begin{center}
\includegraphics[scale=1.0]{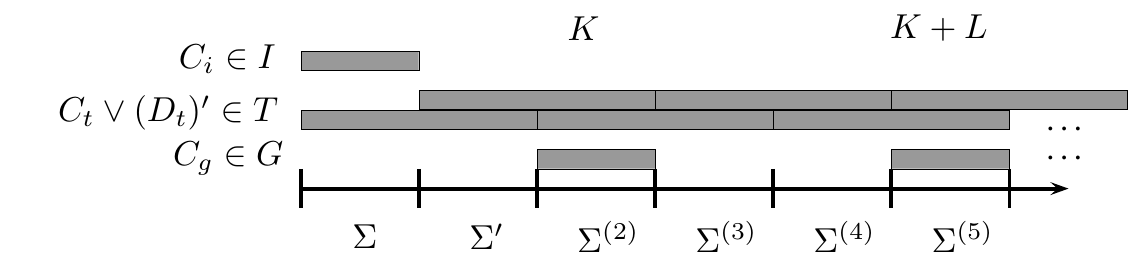}
\caption{Schematic presentation of the potentially infinite set of clauses that is satisfiable if and only if
an LTL-specification $S=(\Sigma,I,T,G)$ has a $(K,L)$-model with $K=2$ and $L=3$.
The axis represents the infinite signature $\Sigma^*$,
while the grey bars stand for individual copies of the initial, step, and goal clauses, respectively.}
\label{fig_klmodel}
\end{center}
\end{figure}

We have now reduced LTL satisfiability of a specification $\mathcal{S}$ to 
infinitely many (for every pair of $K$ and $L$) infinite propositional problems over $\Sigma^*$.
We proceed by assigning labels to the clauses of $\mathcal{S}$ such that a \emph{labelled clause}
represents up to infinitely many \emph{standard clauses} over $\Sigma^*$.
Then an inference performed between labelled clauses corresponds to infinitely many inferences on the level of $\Sigma^*$.
This is similar to the idea of ``lifting'' from first-order theorem proving where clauses with variables represent up to infinitely many ground instances.
Here, however, we deal with the additional dimension of performing infinitely many reasoning tasks on the ``ground level'' in parallel, one for each pair $(K,L)$.
\begin{definition}
A \emph{label} is a triple $(b,k,l) \in \{*,0\}\times(\{*\}\cup \bbbn)\times\bbbn$.
A \emph{labelled clause} $\mathcal{C}$ is a pair $\lclause{(b,k,l)}{C}$
consisting of a label and a standard clause over $\Sigma^*$.
\end{definition}
Semantics of labels is given via a map to certain sets of time indices. 
\begin{definition}
Let $K\in \bbbn$ and $L \in \bbbn^+$ be given. 
We define a set $R_{(K,L)}(b,k,l)$ of indices \emph{represented} by the label $(b,k,l)$
as the set of all $t \in \bbbn$ such that 
\begin{enumerate}
\item
	$b \neq * \rightarrow t = 0$ and	
\item
	$k \neq * \rightarrow \exists s \in \bbbn \, . \, t + k = K+s \cdot L$ and
\item
	$L$ divides $l$.
\end{enumerate}
Now a standard clause of the form $C^{(t)}$ is said to be \emph{represented by the labelled clause
$\lclause{(b,k,l)}{C}$ in $(K,L)$} if $t \in R_{(K,L)}(b,k,l)$.
\end{definition}
The three label components stand for three independent conditions on the time indices to which the clause relates.
The first label component $b$ relates the clause to the beginning of time, and
the second component relates the clause to the indices of the form $K + i\cdot L$, where the goal should be satisfied.
In both cases, $*$ stands for a ``don't care'' value, so if $b$ or $k$ equals $*$, the respective condition is trivially satisfied by any index.
The same effect is achieved for the third condition when $l = 0$, because every positive integer divides $0$.

New label values are computed from old ones using certain operations
when labelled clauses interact in inferences, as will be detailed shortly.
When, initially, a labelled clause set is constructed from an LTL-specification 
(see Definition~\ref{def_initial} below) three particular label values are used.
Further values arise as results of applying the mentioned operations, and the full generality
of labels reflects an entire ``closure'' of the three initial values under these operations.
\begin{definition} \label{def_initial}
Given an LTL-specification $\mathcal{S}=(\Sigma,I,T,G)$, the \emph{initial labelled clause set} $N_\mathcal{S}$ for $\mathcal{S}$ is defined to contain
\begin{itemize}
\item
	labelled clauses of the form $\lclause{(0,*,0)}{C}$ for every $C \in I$,	
\item
	labelled clauses of the form $\lclause{(*,*,0)}{C}$ for every $C \in T$, and	
\item
	labelled clauses of the form $\lclause{(*,0,0)}{C}$ for every $C \in G$.
\end{itemize}
\end{definition}
For any particular choice of $K$ and $L$ the standard clauses over $\Sigma^*$
represented by the labelled clauses from the initial labelled clause set $N_\mathcal{S}$ 
form the purely propositional problem that encodes the existence of a $(K,L)$-model of $\mathcal{S}$. 
\begin{example} 
Let us assume that a specification $\mathcal{S}$ contains a goal clause $(a \lor b) \in G$.
In the initial labelled clause set $N_S$ this goal clause becomes $\lclause{(*,0,0)}{a \lor b}$. 
If we now, for example, fix $K=2$ and $L=3$ as in Fig.~\ref{fig_klmodel}, our labelled clause will represent
all the standard clauses $(a \lor b)^{(t)}$ with $t \in R_{(2,3)}(*,0,0)=\{2,5,8,\dots\}$.
\end{example}


The ultimate goal of this section is to ``lift'' the classical resolution inference rule to labelled clauses.
When two labelled clauses resolve with each other,
a merge operation is applied to their labels to produce the label of the resolvent.
The idea is that the labelled resolvent represents exactly 
those standard clauses that are resolvents of all the possible indicated resolution 
inferences between standard clauses represented by the labelled premises.
\begin{definition}[Labelled resolution]
\begin{equation} \label{labelled_resolution}
\frac{\ \lclause{(b_1,k_1,l_1)}{A \lor C} \quad \lclause{(b_2,k_2,l_2)}{\neg A \lor D}\ }
	   						  {\lclause{(b,k,l)}{C \lor D}} \enspace.
\end{equation}
The two labelled clauses above the line are the inference's premises.
$A$ is an atom, $C$ and $D$ are standard clauses over $\Sigma^*$,
and the label $(b,k,l)$ is the \emph{merge} of $(b_1,k_1,l_1)$ and $(b_2,k_2,l_2)$
defined imperatively as follows:
\begin{itemize}
	\item
		{if} $b_1 = *$ {then} $b := b_2$  
		{else if} $b_2 = *$ {then} $b := b_1$ {else}
		$b := 0$,

	\item	
		{if} $k_1 = *$ {then} $k := k_2$  
		{else if} $k_2 = *$ {then} $k := k_1$ {else}
		$k := \min(k_1,k_2)$,

	\item		
		{if} $k_1 = *$ {or} $k_2 = *$ {then} 
		$l := \gcd(l_1,l_2)$  
		{else} 
		$l := \gcd(l_1,l_2,k_1-k_2)$.
\end{itemize}
\end{definition}
It is straightforward to verify that for every $(K,L)$ the merge operation captures
the \emph{intersection} of the sets of indices represented by its operands
and thus the resulting label represents all the time indices
where standard clauses represented by the inference's premises interact to produce a resolvent. 

\begin{example}
Merge of $(*,2,0)$ and $(*,5,0)$ is $(*,2,3)$; 
we compute the minimum of the $k$ components, and the greatest common divisor of their difference and the original $l$ components.
Merge of $(*,2,3)$ and $(*,2,3)$ is $(*,2,3)$; merge is, in fact, idempotent.
Merge of $(*,2,3)$ and $(*,*,0)$ is $(*,2,3)$; merge has, in fact, a neutral element $(*,*,0)$.
Merge of $(*,2,3)$ and $(0,1,4)$ is $(0,1,1)$.
\end{example}


Not all the resolution inferences from the ``ground level'' of $\Sigma^*$ are directly visible
to the labelled resolution inference (\ref{labelled_resolution}) above.
To obtain a \emph{complete} correspondence, labelled resolution must, in general, be 
preceded by applying the following \emph{time shift} operation
to one of the premises, so that the atom $A$ and its matching partner $\neg A$ from the ``ground level''
become represented by matching counterparts in labelled clauses:
\begin{align}
\lclause{(*,*,l)}{C} & \leadsto \lclause{(*,*,l)}{(C)'}, \label{temp_shift1} \\
\lclause{(*,k,l)}{C} & \leadsto \lclause{(*,k+1,l)}{(C)'}. \label{temp_shift2}
\end{align}
Soundness of time shift is the statement that all the standard clauses represented by the
right hand side of (\ref{temp_shift1}) and (\ref{temp_shift2})
are also represented by the respective left hand sides in any $(K,L)$. 
Note that the operation is \emph{undefined} for labelled clauses with the first component $b=0$,
because these only represent standard clauses fixed to the first time index.




\begin{example}
Let two labelled clauses $\lclause{(*,0,0)}{\neg p \lor q}$ and $\lclause{(*,0,0)}{r \lor p'}$ be given.
They cannot directly participate in a labelled resolution inference, although in $(K,L)=(0,1)$
there are (for every $t$) standard clauses $\neg p^{(t+1)} \lor q^{(t+1)}$ and $r^{(t)} \lor p^{(t+1)}$
represented, respectively, by the two labelled clauses, which resolve on $p^{(t+1)}$.
When the first labelled clause is shifted to $\lclause{(*,1,0)}{\neg p' \lor q'}$,
the clauses resolve on $p'$ and a labelled resolvent $\lclause{(*,0,1)}{r \lor q'}$ is obtained.
\end{example}

\section{Elimination} \label{sec_elim}

By variable and clause elimination we understand the preprocessing technique
described in \cite{EB05} for simplifying propositional SAT problems. 
It consists of a combination of a controlled version of variable elimination
and subsumption\footnote{A standard clause $C$ \emph{subsumes} a clause $D$,
if $C$'s literals are a subset of $D$'s literals. Subsumed clauses are redundant and can be discarded.} reduction for removing clauses, as described below.
These two are alternated in a saturation loop until no further immediate improvement is possible.
This section describes how the mechanism of labelled clauses can be used to 
adapt variable and clause elimination to the context of LTL. 

Propositional variable elimination relies on  \emph{exhaustive} application of the resolution inference rule.
Given (standard) clauses $C = p \lor C_0$ and $D = \neg p \lor D_0$,
their standard resolvent $C \otimes D$ is $C_0 \lor D_0$.
Now, given a propositional problem in CNF consisting of a set of clauses $N$ and a variable $p$,
one separates $N$ into three disjoint subsets $N = N_p \cup N_{\neg p} \cup N_0$ of clauses.
The first set, $N_p$, is a set of clauses containing the variable $p$ positively,
the clauses from $N_{\neg p}$ contain $p$ negatively, and $N_0$ is a set of clauses without variable $p$.
A new clause set $\overline{N}$ is obtained as $(N_p \otimes N_{\neg p}) \cup N_0$, where 
$N_p \otimes N_{\neg p} = \setcompr{C \otimes D}{C \in N_p, D \in N_{\neg p}}.$
The set $\overline{N}$ no longer contains the variable $p$ and is satisfiable if and only if $N$ is.

The obtained set $\overline{N}$ may contain tautological clauses\footnote{A tautological clause contains both a variable 
and its negation.}, which are redundant and should be removed. Then the sizes of $N$ and $\overline{N}$ are compared.
In general, eliminating a single variable may incur a quadratic blowup.
An elimination step is only considered an \emph{improvement} and should be committed to when
the size of $\overline{N}$ is not greater than that of $N$ (possibly up to an additive constant). 
It is shown in \cite{EB05} that improvement eliminations occur often in practice
and that they can be used to simplify the input formula considerably.


Let us now turn to eliminating variables from LTL-specifications.
We know that specifications naturally correspond
to sets of labelled clauses and these in turn represent propositional problems (albeit, in general, infinite ones)
from which variables can be eliminated by the standard procedure described above.
There is still a complication, however, because a single variable $p \in \Sigma$ from the specification
corresponds to all its ``instances'' $p, p', p^{(2)}, \ldots$ on the ``ground level'' of the signature $\Sigma^*$.
To be able to represent the result after elimination, 
all these instances need to be eliminated from the ground level uniformly, in one step.
This seems to be a difficult task when the specification contains a clause
that mentions the variable $p$ in two different time contexts,
like, for example, in $\neg p \lor q \lor p'$.
In this case the individual eliminations cannot be done independently from each other
and we rule the case out from further considerations.

\begin{remark}
There are some interesting subcases where eliminating such a variable would, in theory, be possible and would yield useful results.
Consider the SNF containing $p,\, \always( \neg p \lor p'),\, \always( \neg p \lor r)$,
from which $p$ can be ``semantically''
eliminated and one obtains $\always r$. 
On the other hand, eliminating $p$ from the SNF containing 
$p,\, \always( \neg p \lor  \neg p'),\, \always( p \lor p'),\, \always( \neg p \lor a)$
 should give us a formula whose models $(W_i)_{i \in \bbbn}$ satisfy the condition
$(i \bmod 2 = 0 \Rightarrow W_i \models a)$, which is a property known \cite{Wolper83}
not to be expressible by an LTL formula over the single variable $a$.
\end{remark}

Let us now, therefore, assume that we are given a set of labelled clauses $N$,
perhaps an initial labelled clause set for a specification $\mathcal{S}$,
and a variable $p \in \Sigma$ such that no clause in $N$ contains more 
than one possibly primed occurrence of $p$.
We separate $N$ into $N_p \cup N_{\neg p} \cup N_0$, a subset containing 
$p$ positively (possibly primed), a subset containing $p$ negatively (possibly primed), 
and a subset not containing $p$ at all. A new set of labelled clauses $\overline{N}$
is constructed as $(N_p \otimes N_{\neg p}) \cup N_0$. This time 
$N_p \otimes N_{\neg p}$ stands for the set of all the results of performing labelled resolution inference (\ref{labelled_resolution})
on pairs of clauses from $N_p$ and $N_{\neg p}$, respectively, which may include 
shifting one of the premises in time using the rules (\ref{temp_shift1}) or (\ref{temp_shift2}).

\begin{example} \label{example_eliminate}
Let us assume that a set $N$ contains the following labelled clauses
\begin{align}
\lclause{(0,*,0) &}{p \lor q \lor r},     \label{ex_cl_1} \\
\lclause{(0,*,0) &}{\neg p \lor \neg r},  \label{ex_cl_2} \\
\lclause{(*,*,0) &}{r \lor \neg p'},      \label{ex_cl_3} \\
\lclause{(*,0,0) &}{\neg p \lor q},       \label{ex_cl_4} 
\end{align}
and these are the only labelled clauses of $N$ mentioning variable $p$.
Then eliminating $p$ from $N$ means removing the above labelled clauses 
and replacing them by all the possible labelled resolvents over $p$.
Notice that, actually,
\begin{itemize}
	\item the tautology $(\ref{ex_cl_1})\otimes(\ref{ex_cl_2}) = \lclause{(0,*,0)}{q \lor r \lor \neg r}$ is immediately dropped,
		
	\item and $(\ref{ex_cl_1})\otimes(\ref{ex_cl_3})$ is undefined, because temporal shift does not apply to $(\ref{ex_cl_1})$.	
\end{itemize}	
  Thus the above four clauses are replaced in $N$ by the only nontrivial resolvent
  $(\ref{ex_cl_1})\otimes(\ref{ex_cl_4}) = \lclause{(0,0,0)}{q \lor r}$.
\end{example}

To formulate soundness theorems in this section 
we need a satisfiability notion for labelled clauses.
We extend the definition of a $(K,L)$-model, 
relying on the correspondence between valuations over $\Sigma^*$ and interpretations (see Sect.~\ref{sec_label}).
\begin{definition}
Let $N_{(K,L)} = \setcompr{C^{(t)}}{\lclause{(b,k,l)}{C}\in N \ \&\ t \in R_{(K,L)}(b,k,l)}$ 
denote the set of standard clauses represented in $(K,L)$ by the labelled clauses from $N$.
A set of labelled clauses $N$ is called \emph{$(K,L)$-satisfiable} if there is a valuation $W^*: \Sigma^*\rightarrow \{0,1\}$ which (propositionally) satisfies $N_{(K,L)}$.
The set $N$ is called \emph{satisfiable} if it is $(K,L)$-satisfiable for some $K \in \bbbn$ and $L \in \bbbn^+$.
\end{definition}
Soundness of variable elimination for labelled clauses now reads.
\begin{theorem}
Let $N$ and $\overline{N}=(N_p \otimes N_{\neg p}) \cup N_0$ be sets of labelled clauses as described above.
Then $N$ is $(K,L)$-satifiable if and only if $\overline{N}$ is.
\end{theorem}


Apart from the previously explained limitation, 
there is another restriction on practical variable elimination. 
Consider a clause set consisting of two labelled clauses
$\lclause{(*,*,0)}{\neg x \lor p'}$ and $\lclause{(*,*,0)}{\neg p \lor y'}$.
Eliminating $p$ with the help of labelled resolution yields the single labelled clause $\lclause{(*,*,0)}{\neg x \lor y''}$.
This could be a useful simplification in some contexts, but notice that it got us outside
SNF and LTL-specifications, because $y$ now occurs doubly primed.
There is, nevertheless, an advantage in knowing that such a step can be performed 
(has a proper meaning), because in a more complicated clause set
such a resolvent with undesirable properties might turn out to be redundant 
(for instance, subsumed by another clause)
and would subsequently be removed anyway.

This brings forward the general question of expressivity of labelled clauses.
We know that only the clauses labelled by $(0,*,0), (*,*,0)$ and $(*,0,0)$, which are the labels 
of the initial labelled clause set, directly correspond to initial, step and goal clauses of LTL-specification, respectively.
When clauses with other labels arise during elimination,
the subsequent procedure for deciding satisfiability of the resulting set needs to know how to deal with them.
Interestingly, according to the following theorem, we may drop several kinds of labelled clauses
just after they are created without affecting satisfiability of the clause set. 

\begin{theorem} \label{thm_ignore}
Let $N$ be a finite set of labelled clauses and let $N^-$ be a subset of $N$ obtained be removing 
all the clauses with label of the form $(b,k,l)$ such that either $(b = 0$ and $k \neq *)$ or $(l \neq 0)$.
Then $N^-$ is satisfiable if and only if $N$ is.  
\end{theorem}
\begin{proof}
One implication is trivial as $N^- \subseteq N$. For the other, we need an auxiliary definition.
We say that a label $(b,k,l)$ is \emph{relevant} for a pair $(K,L)$ if $R_{(K,L)}(b,k,l) \neq \emptyset$.
Now any removed clause $\lclause{(b,k,l)}{C}$, i.e. a clause from $N \setminus N^-$, with $(b = 0$ and $k \neq *)$
is only relevant for pairs $(K,L)$ with $K=k$, 
and any removed clause with $(l \neq 0)$ is only relevant for pairs $(K,L)$ with $L$ dividing $l$. 

Let $N^-$ be $(K_0,L_0)$-satisfiable, i.e. some valuation $W^*$ satisfies $(N^-)_{(K_0,L_0)}$.
We may choose $K_1$ of the form $K_0 + i \cdot L_0$ and $L_1$ of the form $j \cdot L_0$ large enough
such that none of the clauses from $N \setminus N^-$ is relevant for $(K_1,L_1)$.
Therefore $(N \setminus N^-)_{(K_1,L_1)} = \emptyset$.
Moreover, $(N^-)_{(K_1,L_1)} \subseteq (N^-)_{(K_0,L_0)}$ by the choice of $K_1$ and $L_1$,
and so $W^*$ satisfies $N_{(K_1,L_1)}$ and thus $N$ is $(K_1,L_1)$-satisfiable. 
\end{proof}

\begin{example}
Deriving an empty labelled clause during elimination does not immediately imply that the current clause set is unsatisfiable.
For instance, the label of the empty clause $\lclause{(*,0,2)}{\emptycl}$ is only relevant for $(K,L)$ when $L$ divides $2$,
and thus the current clause set may still be $(K,L)$-satisfiable for $L > 2$.
\end{example}
After filtering a clause set with the help of Theorem \ref{thm_ignore},
it will only contain clauses with the familiar labels of the initial clause set and
possibly also clauses labelled by $(*,k,0)$, $k \in \bbbn$.
These do not pose any further expressivity complications,
as they arise naturally in our calculus LPSup \cite{LS4PLTLpaper12} for LTL satisfiability.




Let us now turn our focus to reductions, namely to showing how to extend 
subsumption to work with labels.\footnote{Another useful reduction in this context is 
\emph{self-subsuming resolution} \cite{EB05}. It amounts to a resolution inference followed 
by subsumption of one of the premises by the resolvent. Its labelled version
can be derived by combining the presented ideas.} We follow the same idea as with resolution.
Any standard clause represented by the subsumed labelled clause must 
be subsumed by a standard clause represented by the subsuming labelled clause.
Thus we say that $\lclause{(b_1,k_1,l_1)}{C}$ \emph{subsumes} $\lclause{(b_2,k_2,l_2)}{D}$,
if $C$ subsumes $D$ and the merge of the labels $(b_1,k_1,l_1)$ and $(b_2,k_2,l_2)$ is equal to $(b_2,k_2,l_2)$.
Similarly to resolution, the subsumption relation on labelled clauses 
can be made stronger 
if we allow the 
subsuming clause (but not the subsumed one) to be possibly shifted in time.
For example, the clause $\lclause{(*,*,0)}{q}$ subsumes $\lclause{(*,1,0)}{p \lor q'}$ in this sense. 
On the other hand, the clause $\lclause{(*,*,0)}{q'}$ cannot subsume $\lclause{(*,*,0)}{p \lor q}$,
because there is a standard clause represented by the latter, namely $(p \lor q)^{(0)} = p \lor q$,
that is not subsumed by any standard clause represented by the former.
Soundness of labelled clause elimination is stated as follows. 
\begin{theorem}
Let $N$ and $\widetilde{N}$ be sets of labelled clauses, such that $\widetilde{N} \subseteq N$
and for every $\mathcal{D} \in N \setminus \widetilde{N}$
there exists $\mathcal{C} \in \widetilde{N}$ such that $\mathcal{C}$ subsumes $\mathcal{D}$.
Then $N$ is $(K,L)$-satisfiable if and only if $\widetilde{N}$ is.
\end{theorem}
 
We close this section by shortly discussing the overall variable and clause elimination procedure.
As already mentioned, it is advantageous to alternate variable elimination attempts
with exhaustive application of subsumption and possibly other reductions.
That's because removing a subsumed clause may turn elimination
of a particular variable into an improvement and, on the other hand,
new clauses generated during elimination may be subject to subsumption.
This holds true for the original SAT setting as it does with labels.
A detailed description on how to efficiently organize this process can be found in \cite{EB05}.

\section{Experimental evaluation} \label{sec_experiment}


For our evaluation of the effectiveness of variable and clause elimination in LTL,
we extended the preprocessing capabilities of Minisat \cite{minisat} version 2.2.
We kept Minisat's main simplification loop, which efficiently combines variable elimination
with subsumption and self-subsuming resolution, along with the fine-tuned heuristics 
for deciding which variables to eliminate and in what order.
We emulated labels by extending respective clauses with extra \emph{marking} literals\footnote{
For example, any goal clause $C$ is inserted as $C \lor g$, where $g$ is a fresh variable designated for marking goal clauses.}
and, to ensure correctness, we disallowed elimination of variables that occur both primed and unprimed in the input formula.
Although this does not exploit the full potential of variable and clause elimination with labelled clauses
as described in Sect.~\ref{sec_elim},
we already obtained encouraging results with this setup.

For testing we used a set of LTL benchmarks collected by Schuppan and Darmawan \cite{Schuppan11}.
The set consist of total 3723 problems from various sources (mostly previous papers on LTL satisfiability)
and of various flavors (application, crafted, random), and 
represents the most comprehensive collection of LTL problems we are aware of.  
The testing proceeded in three stages. First, all the benchmarks were translated by our tool 
 from the original format into LTL-specifications.
Then we applied the Minisat-based elimination tool 
and obtained a set of simplified LTL-specifications.
Finally, we ran two resolution-based LTL provers 
on both the original and simplified LTL-specifications 
to measure the effect of simplification on prover runtime.
We choose the LTL prover LS4 \cite{pltlProverIJCAR12}, 
most likely the strongest LTL solver\footnote{
LS4 solves 3556 of the above benchmarks within the timelimit of 60s,
the best system reported by Schuppan and Darmawan \cite{Schuppan11}, 
the bounded model checker of NuSMV 2.5, 
is able the solve 3368 of these benchmarks under the same conditions.} currently publically available,
and trp++ \cite{trp03}, a well established temporal resolution prover by Boris Konev.
Having performed the experiments on two independent implementations
should allow us to draw more general conclusions about the effects of variable and clause elimination.

The experiments were performed on our servers with 3.16~GHz Xeon CPU, 16~GB RAM, and Debian~6.0.
All the tools along with intermediate files and experiment logs can be found at \url{http://www.mpi-inf.mpg.de/~suda/vce.html}.

We recorded for each problem the number of variables and clauses 
that we were able to eliminate during the second stage.
We distinguished variables from the \emph{original} problem and \emph{auxiliary} variables
that were introduced during the transformation in stage one.
In total, 39\% of the variables (7\% original, 32\% auxiliary)
and 32\% of the clauses were eliminated. The numbers vary greatly over individual subsets of the benchmarks.
For example, the family \texttt{phltl} allowed for almost no simplification: 
only 3\% of the variables (just auxiliary), and 2\% of the clauses could be removed.
On the other hand, 99\% of the variables (almost all of them original) and 98\% of the clauses
were removed on the family \texttt{O1formula}.
While the former extreme can be explained by a concise and already almost clausal structure
of the original formulas from \texttt{phltl}, the latter follows from the fact
that most of the variables in \texttt{O1formula} occur in just one polarity, i.e. are pure.
Eliminating a pure variable amounts to removal of all the clauses in which the variable appears.\footnote{
If $x$ is a pure variable (literal) then $N_{\neg x}$ is empty and so $N_{x} \otimes N_{\neg x}$ is empty as well.}



\setlength{\tabcolsep}{1em}

\begin{table}[tb]
\caption{
Performance of the two provers on original (o) and simplified (s) problems, grouped by problem subset.
Number of problems \emph{solved} by each prover within the time limit 300 seconds
and the overall \emph{time} spent during the attempts are shown.
Unsolved problems contribute 300.0s, solved at least 0.1s due to the measurement technique.
The times spent on the actual simplification are not included;
these were observed to be negligible for most of the problems,
with maximum of 0.3s for the largest instance.}
\centering
\begin{tabular}{l@{ }r@{\quad\ }c|rr@{.}l|rr@{.}l}
\multirow{2}{*}{subset} & \multirow{2}{*}{
         size} && \multicolumn{3}{c|}{LS4} & \multicolumn{3}{c}{trp++} \\
              \cline{4-9}
        && \quad & \multicolumn{1}{c}{solved} & \multicolumn{2}{r|}{time} & \multicolumn{1}{c}{solved} & \multicolumn{2}{r}{time} \\

\hline
\multirow{2}{*}{\texttt{acacia} } & \multirow{2}{*}{ 
  71}
  & o & 71 
  & 7&1s & 71 
  & 39&3s \\
 && s & 71 
  & 7&1s & 71 
  & 11&3s \\

\hline
\multirow{2}{*}{\texttt{alaska} } & \multirow{2}{*}{ 
  140}
  & o & 121 
  & 6607&0s & 9 
  & 39423&2s \\
 && s & 139 
  & 882&0s & 12 
  & 38717&5s \\

\hline
\multirow{2}{*}{\texttt{anzu}  }& \multirow{2}{*}{ 
  111}
  & o & 93 
  & 5754&2s & 0 
  & 33300&0s \\
 && s & 94 
  & 5482&2s & 0 
  & 33300&0s \\

\hline
\multirow{2}{*}{\texttt{forobots}  }& \multirow{2}{*}{ 
  39}
  & o & 39 
  & 4&3s & 39 
  & 1198&8s \\
 && s & 39 
  & 3&9s & 39 
  & 194&2s \\

\hline
\multirow{2}{*}{\texttt{rozier} }& \multirow{2}{*}{ 
  2320}
  & o & 2278 
  & 13312&9s & 2063 
  & 96293&7s \\
 && s & 2278 
  & 13270&7s & 2120 
  & 76921&1s \\

\hline
\multirow{2}{*}{\texttt{schuppan}  }& \multirow{2}{*}{ 
  72}
  & o & 41 
  & 9332&8s & 36 
  & 11189&8s \\
 && s & 41 
  & 9320&9s & 37 
  & 10741&0s \\

\hline
\multirow{2}{*}{\texttt{trp} }& \multirow{2}{*}{ 
  970}
  & o & 940 
  & 12327&5s & 364 
  & 189045&2s \\
 && s & 934 
  & 11887&5s & 359 
  & 190138&3s \\

\hline \hline
\multirow{2}{*}{total  }& \multirow{2}{*}{ 
  3723}
  & o & 3583 
  & 47345&8s & 2582 
  & 370490&0s \\
 && s & 3596 
  & 40854&3s & 2638 
  & 350023&4s \\
\end{tabular}
\label{tab_experiment}
\end{table}

The results of the third stage, in which we measured the effect of simplification on the performance of the two selected provers, are
summarized in Table~\ref{tab_experiment} and at the same time represented graphically in Fig.~\ref{fig_runtimes}.
We see that both LS4 and trp++ substantially benefit from the simplification, both in the number of solved instances and the overall runtime.
On some subsets the effect is quite pronounced (see, e.g., LS4 on \texttt{alaska} or trp++ on \texttt{forobots}), while on others it is more modest.
Only on the subset \texttt{trp} did the simplification result in less problems solved. What the table does not show, however, 
is that even among the \texttt{trp} problems there were some only solved in the simplified form (16 such problems for LS4 and 9 for trp++).
When judging the relative number of problems gained by each prover,
it should be noted that many problems come from scalable families and 
are mostly trivial or too difficult to solve. This leaves the ``grey zone''
where improvement is possible relatively small.

\begin{figure}[tb]
\mbox{
\hspace{-0.3cm}
\subfigure{\includegraphics{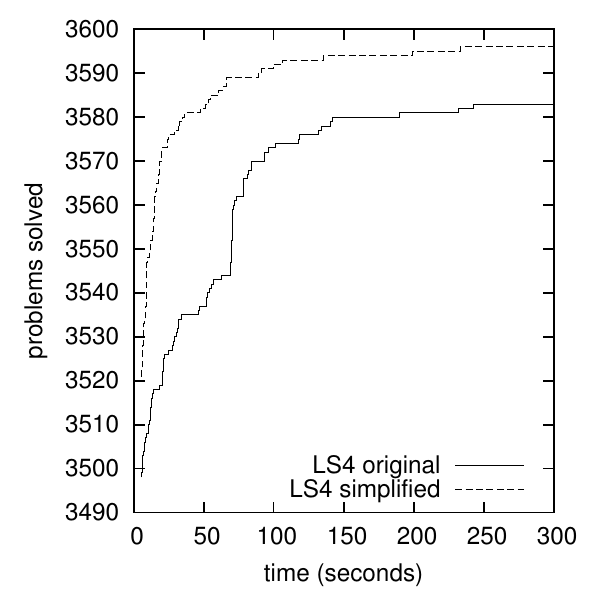}}
\subfigure{\includegraphics{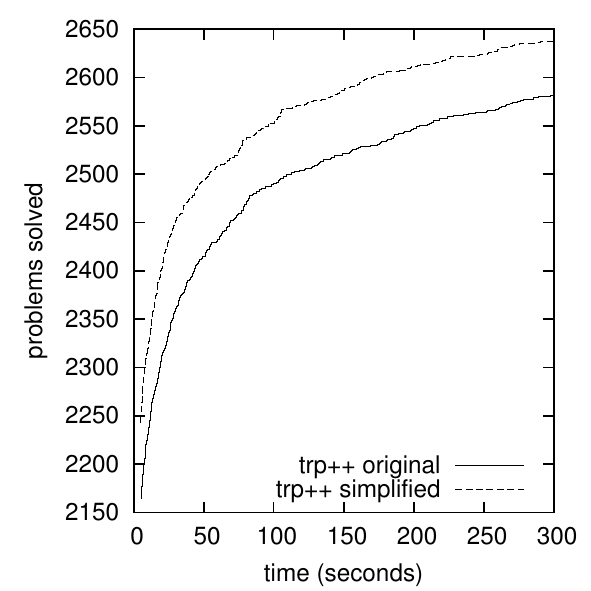}}
}
\caption{
Comparing the number of problems solved, simplified and original, within a given time limit.
Although the value ranges for LS4 (on the left) and trp++ (on the right) differ, both figures
demonstrate better performance on the simplified problems.}
\label{fig_runtimes}
\end{figure}

To conclude, the result of our evaluation indicate that variable and clause elimination
represents a useful preprocessing technique of LTL-specifications.
Simplifying a clause set not only removes redundancies introduced by a previous, potentially sub-optimal 
normal form transformation (when auxiliary variables get eliminated),
but usually reduces the input even further.
This ultimately decreases the time needed to solve the problem.
Further improvements are expected from an independent implementation
that will harness the full potential of the mechanism of labels.

\section{Discussion} \label{sec_discuss}


We are not aware of any related work directly focusing on simplifying clause normal forms for LTL.
However, some interesting connections can be drawn with the help of Remark \ref{rem_buchi} of Sect.~\ref{sec_normal},
which shows that an LTL-specification can be viewed as a symbolic representation of a \buchi\ automaton.
For instance, in the classical paper \cite{GPVW95}, an automaton accepting the models of an LTL formula $\varphi$
is constructed such that its states are identified with sets of $\varphi$'s subformulas.
A closer look reveals an immediate connection between these subformulas and
the variables introduced to represent them in the SNF for $\varphi$.
The above paper also suggests several improvements of the basic algorithm.
For instance, it is advocated that subformulas of the form $\mu_1 \land \mu_2$ need not be stored,
because the individual conjuncts $\mu_1$ and $\mu_2$ will be later added as well
and they already imply the conjunction as a whole. We can restate this 
on the symbolic level as an observation that a variable introduced
to represent a conjunctive subformula can always be eliminated, which is a claim easy to verify.

We believe this connection deserves further exploration,
as one could possibly use it to bring some of the numerous techniques
for optimizing explicit automata construction (see e.g. \cite{automRozVar07}) to the symbolic level. 
Note, however, that the main application of the explicit automata construction approach
lies in model checking and so the resulting automaton 
is required to be \emph{equivalent} to the original formula.
On the other hand, our clausal symbolic approach is meant for satisfiability testing only
and so more general \emph{satisfiability preserving} transformations are allowed.
An elimination of a variable from the original signature of the formula $\varphi$,
or the ``forgetting step'' justified by Theorem~\ref{thm_ignore} of Sect.~\ref{sec_elim},
are examples of transformations that do not have a counterpart on the automata side.

While the explicit notion of a symbolic representation of a \buchi\ automaton via a clause normal form
has received relatively little attention so far\footnote{A correspondence between SNF and \buchi\ automata has been shown
in \cite{BolotovEtal02}. The relevant theorem of the paper, however, does not establish an equivalence
between  models of the formula and accepting runs of the automaton. Its value for translating 
techniques between the symbolic and explicit approaches is, therefore, limited.}, 
symbolic approaches to LTL model checking and satisfiability based on Binary Decision Diagrams 
are well known \cite{Clarke97anotherlook}. Again, it seems possible that some optimization techniques
could be shared between the two approaches. For instance, different BDD encodings
recently studied by Rozier and Vardi \cite{multiRozVar11}, could correspond to different ways 
of turning a formula into an LTL-specification.

\section{Conclusion} \label{sec_concl}

We have shown that variable and clause elimination,
a practically successful preprocessing technique for propositional SAT problems,
can be adapted to the setting of linear temporal logic.
For that purpose we have utilized the mechanism of labelled clauses,
a method for interpreting an LTL formula as finitely represented 
infinite sets of standard propositional clauses.
The ideas were implemented and tested on a comprehensive set of benchmarks with encouraging results.
In particular, variable and clause elimination has been shown to significantly improve
subsequent runtime of resolution-based provers LS4 and trp++.

We would like to stress here that labelled clauses provide a general method 
for transferring resolution-based reasoning from SAT to LTL.
It is therefore plausible that other preprocessing techniques,
like, for example, the blocked clause elimination \cite{BlockedClause10},
can be adapted along the same lines. Exploring this possibility will be one of the directions for future work.


\bibliographystyle{abbrv}
\bibliography{clausalLTL}

\newpage
\appendix
\section{LTL preliminaries} \label{sec_prelim}

The language of Linear Temporal Logic (LTL) formulas is an extension of the propositional language
with temporal operators. The most commonly used are Next $\ltlnext$, Always $\always$,
Eventually $\eventually$, Until $\until$, and Release $\release$.
Formally, let $\Sigma = \{p,q,\ldots\}$ be a (finite) signature of propositional variables,
then the set of LTL formulas is defined inductively as follows:
\begin{itemize}
\item
	any $p \in \Sigma$ is a formula,

\item
	if $\varphi$ and $\psi$ are formulas, then so are $\neg \varphi$, $\varphi \land \psi$, and $\varphi \lor \psi$,

\item
	if $\varphi$ and $\psi$ are formulas, then so are $\ltlnext \varphi$, $\always \varphi$, $\eventually \varphi$, $\varphi \until \psi$, and $\varphi \release \psi$.

\end{itemize}

A propositional valuation, or simply a \emph{state}, is a mapping $W : \Sigma \rightarrow \{0,1\}$.
An \emph{interpretation} for an LTL formula is an infinite sequence of states $\mathcal{W} = (W_i)_{i \in \bbbn}$.
The truth relation $\mathcal{W},i \models \varphi$ between an interpretation $\mathcal{W}$, time index $i \in \bbbn$, and a formula $\varphi$
is defined recursively as follows:
\begin{flushleft}
\hspace{0.1cm}\begin{tabular}{l@{\quad}l}		
		$\mathcal{W},i \models p$ & iff $W_i \models p$, \\
		$\mathcal{W},i \models \neg \varphi$ & iff not $\mathcal{W},i \models \varphi$, \\
		$\mathcal{W},i \models \varphi \land \psi$ & iff $\mathcal{W},i \models \varphi$ and $\mathcal{W},i \models \psi$, \\
		$\mathcal{W},i \models \varphi \lor \psi$ & iff $\mathcal{W},i \models \varphi$ or $\mathcal{W},i \models \psi$, \\						
	  $\mathcal{W},i \models \ltlnext \varphi$ & iff $\mathcal{W},i+1 \models \varphi$, \\	  	  
	  $\mathcal{W},i \models \always \varphi$ & iff for every $j \geq i$, $\mathcal{W},j \models \varphi$, \\
	  $\mathcal{W},i \models \eventually \varphi$ & iff for some $j \geq i$, $\mathcal{W},j \models \varphi$, \\				
  \end{tabular}
\end{flushleft}	

\begin{center}
\begin{tabular}{l@{\quad}l}	  		
		$\mathcal{W},i \models \varphi \until \psi$ & iff there is $j \geq i$ such that $\mathcal{W},j \models \psi$ \\
		 & and $\mathcal{W},k \models \varphi$ for every $k$, $i\leq k < j$, \\
		$\mathcal{W},i \models \varphi \release \psi$ & iff for all $j \geq i$, $\mathcal{W},j \models \psi$ or \\
		 & there is $j \geq i$ with $\mathcal{W},j \models \varphi$ and for all $k$, $i\leq k \leq j$, $\mathcal{W},k \models \psi$.  \\
\end{tabular}
\end{center}		 
An interpretation $\mathcal{W}$ is a \emph{model} of an LTL formula $\varphi$ if $\mathcal{W},0 \models \varphi$. A formula $\varphi$ is called \emph{satisfiable}
if it has a model, and is called \emph{valid} if every interpretation is a model of $\varphi$.

\section{Transforming LTL formulas to SNF} \label{sec_ltl2snf}

Formulas in SNF are conjunctions of \emph{temporal clauses}, 
each of them assuming one of the following forms:
\begin{itemize}
\item
	an \emph{initial} clause: $\bigvee_j k_j$,

\item	
	a \emph{step} clause: $\always (\bigvee_j k_j \lor \bigvee_j \ltlnext l_j)$,
	
\item	
	an \emph{eventuality} clause: $\always (\bigvee_j k_j \lor \eventually l)$,

\end{itemize}
where $k_j, l_j,$ and $l$ stand for standard literals, i.e. propositional variables or their negation.

The translation of an LTL formula $\varphi$ into an equisatisfiable SNF starts by first turning 
$\varphi$ into an equivalent formula that is in Negation Normal Form (NNF), meaning the negation sign
only occurs in front of propositional variables in the leaves of the formula tree.
This can be achieved by a standard operation that ``pushes negations downwards'' with the help of De Morgan's rules
and temporal equivalences like $\neg \ltlnext \varphi \equiv \ltlnext \neg \varphi$, $\neg \always \varphi \equiv \eventually \neg \varphi$, 
and $\neg (\varphi \until \psi) \equiv (\neg \varphi) \release (\neg \psi)$.
Finally, multiple negations are absorbed with the help of the classical equivalence
$\neg \neg \varphi \equiv \varphi$. In what follows we assume that $\varphi$ is already in NNF.

\begin{figure}
\begin{center}
\begin{tabular}{l@{$\ \ \ $}r@{$\ \ \ $}c@{$ \ \ \ $}l}

	1. & $\tau[\always (\neg x \lor l)]$ & $\longrightarrow$ & $\always (\neg x \lor l)$, if $l$ is a literal, \\
	
  \noalign{\smallskip}
  2. & $\tau[\always (\neg x \lor (\varphi \land \psi))]$ & $\longrightarrow$ &
                                                      $\tau[\always (\neg x \lor \varphi)] \land
                                                       \tau[\always (\neg x \lor \psi)],$ \\
                                                        
  \noalign{\smallskip}
  3. & $\tau[\always (\neg x \lor (\varphi \lor \psi))]$ & $\longrightarrow$ &
                                                       $\always (\neg x \lor \mathbf{u} \lor \mathbf{v}) \land $\\
    &                                                & & $\tau[\always (\neg \mathbf{u} \lor \varphi)] \land
                                                       \tau[\always (\neg \mathbf{v} \lor \psi)],$ \\
                                                     
  \noalign{\smallskip}
  4. & $\tau[\always (\neg x \lor \ltlnext \varphi )]$ & $\longrightarrow$ & 
                                                     $\always (\neg x \lor \ltlnext \mathbf{u}) \land $ \\
    &                                             & & $\tau[\always (\neg \mathbf{u} \lor \varphi)],$ \\
                                                     
  \noalign{\smallskip}
  5. & $\tau[\always (\neg x \lor \always \varphi )]$ & $\longrightarrow$ &
                                                    $\always (\neg x \lor \mathbf{u}) \land
                                                     \always (\neg \mathbf{u} \lor \ltlnext \mathbf{u}) \land $ \\
   &                                              & & $\tau[\always (\neg \mathbf{u} \lor \varphi)],$ \\
                                                    
  \noalign{\smallskip}
  6. & $\tau[\always (\neg x \lor \eventually \varphi )]$ & $\longrightarrow$ &
                                                        $\always (\neg x \lor \eventually \mathbf{u}) \land $ \\
   &                                                  & & $\tau[\always (\neg \mathbf{u} \lor \varphi)],$ \\  
                                                        
	\noalign{\smallskip}
  7. & $\tau[\always (\neg x \lor (\varphi \until \psi )]$ & $\longrightarrow$ &
                                                        $\always (\neg x \lor \eventually \mathbf{v}) \land $ \\
   &                                                 & & $\always (\neg x \lor \mathbf{v} \lor \mathbf{w}) \land
                                                         \always (\neg \mathbf{w} \lor \mathbf{u}) \land
                                                         \always (\neg \mathbf{w} \lor \ltlnext \mathbf{v} \lor \ltlnext \mathbf{w}) \land$ \\
   &                                                & & $ \tau[\always (\neg \mathbf{u} \lor \varphi)] \land
                                                         \tau[\always (\neg \mathbf{v} \lor \psi)],$ \\                                                    
  
  \noalign{\smallskip}
  8. & $\tau[\always (\neg x \lor (\varphi \release \psi )]$ & $\longrightarrow$ &
                                                        $\always (\neg x \lor \mathbf{w}) \land 
                                                         \always (\neg \mathbf{w} \lor \mathbf{v}) \land  
                                                         \always (\neg \mathbf{w} \lor \mathbf{u} \lor \ltlnext \mathbf{w}) \land$ \\                                                   
   &                                                & & $ \tau[\always (\neg \mathbf{u} \lor \varphi)] \land
                                                         \tau[\always (\neg \mathbf{v} \lor \psi)],$ \\

\end{tabular}

\caption{The rules for SNF transformation. The freshly introduced variables are in bold.}
\label{fig_taurules}
\end{center}
\end{figure}


The actual transformation is performed with the help of operator $\tau$ defined in Fig.~\ref{fig_taurules},
which recursively reduces any formula of the form $\always (\neg x \lor \varphi)$ into the final SNF.
During the process, new ``fresh'' variables are being introduced (we typeset them in bold) which serve two different purposes:
They stand as names for subformulas (as in the case of the rules for, e.g., conjunction),
and may also play a role of ``trackers'' that influence the value of other variables 
not just in the current state, but also in those to follow. This is how the semantics of, e.g., 
the Always operator $\always$ is being encoded. The overall translation is triggered by the following rule
\[\varphi \ \ \ \longrightarrow \ \ \  \mathbf{i} \land \tau[\always ( \neg \mathbf{i} \lor \varphi)] \enspace, \]
with a fresh variable $\mathbf{i}$ that represents the whole formula.

\begin{example}
Here we work out an example from \cite{ctrFDP01} to demonstrate the translation procedure.
Assume we would like to prove the formula $(\eventually p \land \always (p \rightarrow \ltlnext p))\rightarrow \eventually \always p$.
In refutational theorem proving we proceed by negating the formula and trying to show the negation to be unsatisfiable.
By taking the negation into NNF (and translating away the implication symbol) we obtain
\[ (\eventually p \land \always (\neg p \lor \ltlnext p)) \land \always \eventually \neg p \enspace,\]
which is consequently translated into the following set of clauses:
\begin{center}
\begin{tabular}{c@{$\ \ \ $}l}
$i$ & By the initial rule. \\
  \noalign{\smallskip}
$\always (\neg i \lor \eventually u_1) $ & The first conjunct by rule 6, \\
$\always (\neg u_1 \lor p) $ & terminates by rule 1. \\
  \noalign{\smallskip}
$\always (\neg i \lor u_2) $ \\
$\always (\neg u_2 \lor \ltlnext u_2)$ & The second conjunct by rule 5, \\
$\always (\neg u_2 \lor u_3 \lor v_3)$ & inside which there is disjunction (rule 3), \\
$\always (\neg u_3 \lor \neg p)$ & the first argument is a literal (rule 1), \\
$\always (\neg v_3 \lor \ltlnext u_4)$ & the second goes by rule 4 \\
$\always (\neg u_4 \lor p) $ & and terminates by rule 1. \\
  \noalign{\smallskip}
$\always (\neg i \lor u_5) $ \\ 
$\always (\neg u_5 \lor \ltlnext u_5)$ & The third conjunct by rule 5, \\
$\always (\neg u_5 \lor \eventually u_6)$ & inside which we apply rule 6, \\
$\always (\neg u_6 \lor \neg p)$ & and terminate by rule 1. 
\end{tabular}
\end{center}
Notice that transformation $\tau$ introduces more new variables than would be strictly necessary.
For example, the variable $u_6$ just ``connects'' the last two clauses, which could be replaced by 
one equivalent eventuality clause $\always (\neg u_5 \lor \eventually \neg p)$. 
This is a price we pay here for the simple statement of the transformation rules in Fig.~\ref{fig_taurules}
(no side conditions
). An actual implementation would strive to detect the literal case as soon as possible,
and thus, e.g., introduction of $u_6$ would be avoided.
\end{example}


\section{Transforming general SNF to LTL-specification} \label{sec_snf2spec}

The transformation of general SNF to LTL-specifications focuses on eventuality clauses.
It consists in two simplification steps:
\begin{enumerate}
\item
	turning the \emph{conditional} eventuality clauses into \emph{unconditional} ones (of the form $\always \eventually l$),
\item
  reducing \emph{multiple} (unconditional) eventuality clauses from the SNF into \emph{just one} eventuality clause.
\end{enumerate}
We present our modification of the simplifications first introduced in \cite{simplifiedDFK02} that performs both steps at once.

Assume that an SNF of a formula contains $n$ (in general) conditional eventuality clauses
\[\always (C_i \lor \eventually l_i)\]
for $i=1,\ldots,n$, where $C_i$ is the conditional part, i.e.~a disjunction of literals.
We remove these, and replace them with a single unconditional eventuality clause
\begin{equation}
\always \eventually \mathbf{m} \label{cl_0}
\end{equation}
together with the following five step clauses for every $i=1,\ldots,n:$ 
\begin{align}
\always & ( C_i \lor l_i  \lor \mathbf{t_i} ), \label{cl_1} \\
\always & ( \neg \mathbf{t_i} \lor \ltlnext l_i  \lor \ltlnext \mathbf{t_i} ), \label{cl_2} \\
\always & ( \mathbf{s_i} \lor \neg \mathbf{t_i}  \lor \ltlnext \neg \mathbf{s_i}),  \label{cl_3} \\
\always & ( \neg \mathbf{s_i} \lor  \neg \mathbf{m}), \label{cl_4} \\
\always & ( \mathbf{s_i} \lor \ltlnext \neg \mathbf{m}), \label{cl_5}
\end{align}
where again the bold variables are supposed to be new to the formula.

The idea behind the simplification is the following: 
If the condition $\neg C_i$ is satisfied in the current state and the respective eventuality $l_i$ is not satisfied in the same state
we start ``tracking'' the eventuality  with the help of the new variable $\mathbf{t_i}$ (clause \ref{cl_1}).
The tracking variable $\mathbf{t_i}$ is forced to stay true also in the future states unless the eventuality $l_i$ is 
finally satisfied (clause \ref{cl_2}). Now let us look from the other side. The unconditional eventuality (clause \ref{cl_0})
will infinitely often ensure that all the variables $\mathbf{s_i}$ are false in one state (clause \ref{cl_4})
and were true in the previous state (clause \ref{cl_5}). Thus in the intervals between states where $\mathbf{m}$ holds,
there will always be two consecutive states where $\mathbf{s_i}$ changes from false to true. But this cannot happen
if we are tracking that particular eventuality at that time (clause \ref{cl_3}). 
To sum up, for each of the original eventualities we have a guarantee that 
in every interval between states where $\mathbf{m}$ holds
the eventuality was either not triggered at all ($\neg C_i$ was false in the whole interval)
or the eventuality was triggered and subsequently satisfied in that interval.
Please consult \cite{simplifiedDFK02} for a formal proof.

\begin{example} \label{example_multi_single}
Our previous example contained two conditional eventuality clauses
$\always (\neg i \lor \eventually u_1) $ and $\always (\neg u_5 \lor \eventually u_6)$.
We may replace these by the following set of clauses
to obtain an equisatisfiable problem with just one unconditional eventuality clause:
\begin{center}
\begin{tabular}{c@{$\ \ \ $}l}
$\always \eventually m$, \\
$\always  ( \neg i \lor u_1  \lor t_1 )$, \\
$\always  ( \neg t_1 \lor \ltlnext u_1  \lor \ltlnext t_1 )$, \\
$\always  ( s_1 \lor \neg t_1  \lor \ltlnext \neg s_1)$,\\
$\always  ( \neg s_1 \lor  \neg m)$,  \\
$\always  ( s_1 \lor \ltlnext \neg  m)$, \\
$\always  ( \neg u_5 \lor u_6  \lor t_2 )$, \\
$\always  ( \neg t_2 \lor \ltlnext u_6  \lor \ltlnext t_2 )$, \\
$\always  ( s_2 \lor \neg t_2  \lor \ltlnext \neg  s_2)$, \\
$\always  ( \neg s_2 \lor  \neg m)$, \\
$\always  ( s_2 \lor \ltlnext \neg  m)$. \\
\end{tabular}
\end{center}
\end{example}






\end{document}